\documentclass[pra,two column,usletter]{revtex4-1} % for arxiv
\usepackage{amssymb,amsmath,bm} 
\usepackage{mathrsfs}

\makeatletter
\def \v#1{{\bm #1}}
\def \be {\begin{equation}}
\def \ee {\end{equation}}
\def \bml{\begin{multline}}
\def \eml{\end{multline}}
\newcommand{\Exp}[1]{\,\mathrm{e}^{\mbox{\footnotesize$#1$}}}
\newcommand{\I}{\mathrm{i}}

\newcommand{\ket}[1]{|#1\rangle}
\newcommand{\bra}[1]{\langle#1|}
\newcommand{\tr}[1]{\mathrm{tr}\left(#1\right)}
\newcommand{\Tr}[1]{\mathrm{Tr}\left\{#1\right\}}
\newcommand{\Trabs}[1]{\mathrm{TrAbs}\left\{#1\right\}}

\def \ds{\displaystyle}

\newcommand{\vecin}[2]{\langle#1|#2\rangle}
\newcommand{\sldin}[2]{\langle#1,#2\rangle_{\rho_{\theta}}}
\newcommand{\rldin}[2]{\langle#1,#2\rangle_{\rho_{\theta}}^{+}}
\def \del{\partial}
\def \stheta{\v{s}_{\theta}}

\def \cB{{\cal B}}

\def \cM{{\cal M}}
\def \cH{{\cal H}}
\def \cX{{\cal X}}
\def \sofh{{\cal S}({\cal H})}

\def \bbr{{\mathbb R}}
\def \bbc{{\mathbb C}}

\def \sofc2{{\cal S}({\mathbb C}^2)}

\def \lofh{{\cal L}({\cal H})}
\def \lofhh{{\cal L}_h({\cal H})}

\newtheorem{theorem}{Theorem}[section]
\newtheorem{lemma}[theorem]{Lemma}
\newtheorem{proposition}[theorem]{Proposition}

\newenvironment{proof}[1][Proof:]{\begin{trivlist}
\item[\hskip \labelsep {\bfseries #1}]}{\end{trivlist}}

\newcommand{\qed}{\nobreak \ifvmode \relax \else
      \ifdim\lastskip<1.5em \hskip-\lastskip
      \hskip1.5em plus0em minus0.5em \fi \nobreak
      \vrule height0.75em width0.5em depth0.25em\fi}

\begin{document}

%%%% For arxiv
\title{Parameter estimation of qubit states with unknown phase parameter}
\author{Jun Suzuki\\
(junsuzuki@is.uec.ac.jp)}
\date{\today}
%\email{junsuzuki@is.uec.ac.jp}
\affiliation{
Graduate School of Information Systems, The University of Electro-Communications,\\
1-5-1 Chofugaoka, Chofu-shi, Tokyo, 182-8585 Japan
}

\begin{abstract}
We discuss a problem of parameter estimation for quantum two-level system, qubit system, in presence of 
unknown phase parameter. 
We analyze trade-off relations for mean-square errors when estimating relevant parameters with separable 
measurements 
based on known precision bounds; the symmetric logarithmic derivative Cramer-Rao bound 
and Hayashi-Gill-Massar (HGM) bound. 
We investigate the optimal measurement which attains the HGM bound and discuss its properties. 
We show that the HGM bound for relevant parameters can be attained asymptotically 
by using some fraction of given $n$ quantum states to estimate the phase parameter. 
We also discuss the Holevo bound which can be attained asymptotically by a collective measurement. 
\end{abstract}

%\keywords{Quantum parameter estimation, nuisance parameter, trade-off relation}

\maketitle %For arxiv
%=====================================================================
\section{Introduction}
Quantum statistical inference is of fundamental importance 
not just from foundation of quantum information theory but also in view of 
practical applications. For example, at a certain stage of any quantum information processing 
protocol, one has to know the state precisely to proceed the protocol. 
Typically, the quantum states to be estimated are not completely unknown, 
but we have partial information about them. This is contrast to quantum tomography 
where one has to identify a quantum state by informationally complete measurements. 

Quantum parameter estimation problem, which is a subclass of quantum statistical inference problems, 
assumes that a given quantum state is parameterized with a finite number of continuous parameters. 
One wishes to infer the value of these parameters by performing a measurement and 
making an estimate from measurement outcomes. Parameter estimation problem in 
classical statistics is a well-established subject and there are large numbers of literature 
available ranging from rigorous mathematical formulation to very practical applications. 
Quantum parameter estimation was initiated by Helstrom in 60s and 
developed by Holevo, Yuen-Lax, and others \cite{helstrom, holevo, yl73}. 
The new insight into this problem was triggered by Nagaoka in the late 80s 
where he developed new language based on information geometry in classical statistics \cite{ANbook} 
and opened asymptotical analysis of estimation. Some of his contributions are reprinted in Ref.~\cite{hayashi}. 
The field of quantum estimation theory has recently gained great attentions also 
from physics community. One important motivation is the study of 
quantum metrology, that is, high precision measurement  which 
go beyond existing classical precision limit \cite{glm11}. 

The aim of this paper is to discuss some of unexplored aspect of 
quantum parameter estimation problem. We analyze an estimation problem 
in presence of unknown parameter, called a {\it nuisance parameter} in statistics, and discuss 
effects of this nuisance parameter. This problem is well-known in classical statistics \cite{lc,bnc}, 
yet few results are known in quantum case. For this purpose, we take the simplest 
quantum system, a qubit system, and we apply known precision bound to our estimation problem. 
We see that effects of nuisance parameters are important in general. 
For a very special case, asymptotically achievable bound can be 
obtained as shown in this paper.  

A quantum parametric model studied in this paper is 
\[
\rho_{\theta}=\frac{1}{2}\left(\begin{array}{cc}1+\theta_2 & \theta_1 \Exp{-\I\theta_3} \\[.5ex] 
\theta_1 \Exp{\I\theta_3} & 1-\theta_2\end{array}\right), 
\]
where the parameters $\theta_1,\theta_2$ are of interest and 
the phase parameter in the off-diagonal component is not.  
This model is physically motivated from wave-particle duality, 
where one discusses the trade-off between the fringe visibility $\Leftrightarrow\ |\tr{\rho_{\theta}\sigma_+}|=|\theta_1|$ 
and the which-way information $\Leftrightarrow\ |\tr{\rho_{\theta}\sigma_3}|=|\theta_2|$, whereas 
the value of phase $\theta_3$ itself is irrelevant. 
We should not forget to mention several works related to the present work. 
Similar parameter models for mixed qubit states was discussed by several authors. 
Among them, Gill and Massar derived a very general trade-off relation known as 
the Gill-Massar inequality and derived an achievable bound for qubit case \cite{GM00}. 
Bagan {\it et al} studied two and three parameter model with different parametrization 
and different figures of merit \cite{bbgmm06}. 
Hayashi and Matsumoto performed a general analysis on asymptotic performance in 
qubit system and analyzed two and three parameter models.  

Our model is different from the previous results in three aspects. 
Firstly, parametrization is different and we do not use neither 
cartesian nor spherical coordinates in the Bloch vector representation as were analyzed in literature. 
Secondly, we shall not assume the value of phase is known. 
When this value is completely known, the model is reduced 
to two-parameter model which lies on an equatorial plane of the Bloch sphere. 
In contrast, we are interested in analyzing errors in estimating two parameters 
without knowing the value of phase. Lastly, in many studies on quantum metrology, 
one is interested in estimating the value of phase and the amplitude damping 
or phase dephasing caused by external noise are not \cite{emfd11,ddkg12}. 
In recent publications \cite{cdbw14,vdgjkkdbw14}, authors point out 
that we cannot estimate the value of phase in presence of noise in particular 
noise model. They derive a trade-off relation between 
error in these parameters together with experimental demonstration.  
Instead, this paper aims to shed light on parameter estimation problem 
in presence of unknown parameter based on quantum parameter estimation perspective. 

This paper is organized as follows. Section \ref{sec2} provides 
a brief summary of basic theorems in quantum estimation theory. 
Section \ref{sec3} discusses our parametric model and quantum 
estimation in presence of unknown phase parameter. 
Section \ref{sec4} shows the ultimate bound based on the Holevo bound. 
We also discuss the general structure of quantum estimation with nuisance parameters. 
We close the paper with conclusion and outlook in Section \ref{sec5}. 

\section{Preliminaries}\label{sec2}
In this section we summarize definitions for basic terminologies and quantities to 
analyze quantum parametric models. Previously known results are listed without proofs.  
For more details, readers are referred to books \cite{holevo, ANbook, hayashi, petz} and 
the concise summary by Hayashi and Matsumoto \cite{HM08}. 

\subsection{Estimation problems and Fisher information in classical and quantum cases}
Let $\cH$ be a finite dimensional complex Hilbert space and $\lofh$ denote 
the set of all linear operator from $\cH$ to itself. A quantum state $\rho$ 
is an element of $\lofh$, which is non-negative and has a unit trace. The 
totality of quantum states on $\cH$ is written as $\sofh:= \{\rho\in\lofh \, |\, \rho\ge 0, \tr\rho=1\}$. 
A measurement on a given quantum state $\rho$ is described by 
a positive operator-valued measure (POVM) or probability operator measurement, 
which is a set of non-negative operators summed up to the identity operator on $\cH$. 
In this paper, we shall only consider discrete POVMs whose elements are countable. 
We denote the label set by $\cX$. 
The POVM is expressed as $\Pi=\{ \Pi_x\in\lofh\,|\, \Pi_x\ge0,\sum_{x\in\cX}\Pi_x=I\}_{x\in\cX}$. 
%The set of all POVMs on $\cH$ is denoted as $\mathscr{M}(\cH):=\{\Pi\,|\,\Pi$ is a POVM on $\cH\}$. 
One of the axioms of quantum mechanics (Born's interpretation) provides a simple rule: 
the probability distribution for detecting the measurement outcome $x$ when 
a POVM $\Pi$ is performed for a given quantum state $\rho$ is $p_\rho(x)=\tr{\rho\Pi_x}$. 
This is a conditional probability distribution and the condition $\rho$ is omitted when is clear from the context. 

A quantum parametric model is given as a family of quantum states on $\cH$, 
which is parametrized by a $k$-dimensional parameter $\theta=(\theta_1,\theta_2,\dots,\theta_k)\in\bbr^k$ and is denoted by 
\be
\cM_Q=\{\rho_\theta\,|\,\theta\in\Theta  \}. 
\ee
Here the parameter set $\Theta$ is assumed to be an open subset of $\bbr^k$ and 
we also assume that $\rho_\theta$ varies smoothly enough so that no singular 
behaviors for information quantities defined later. 
Given a quantum model $\cM_Q$, the aim of quantum statistician is two-fold: 
First she performs a good measurement $\Pi$  
and then makes a good estimate $\hat\theta=(\hat\theta_1,\hat\theta_2,\dots,\hat\theta_k)$ based on her measurement outcomes. 
The quality of her estimation is measured according to a given figure of merit 
such as the mean square error, minimax error, Bayesian criterion, etc. 
In the following discussion, we choose the figure of merit as the mean square error (MSE)  defined by 
\be
v_{\theta,ij}[\Pi,\hat{\theta}] := \sum_{x\in\cX} (\hat{\theta}_i(x)-\theta_i) (\hat{\theta}_j(x)-\theta_j)\tr{\rho_{\theta}\Pi_x}, 
\ee
where the indices are $i,j=1,2,\dots,k$ and the $k\times k$ real matrix 
$V_{\theta}[\Pi,\hat{\theta}]:=[v_{\theta,ij}]_{i,j\in\{1,2,\dots,k\}}$ 
is called a MSE matrix. It is straightforward to see the MSE matrix is symmetric and non-negative matrix. 
The second process above is the same as (classical) statistics and is described 
by a function $\hat{\theta}$: $\cX\rightarrow \bbr^k$ 
(In general, the range satisfies $\hat{\theta}(\Theta)\supset \Theta$. 
We take $\hat{\theta}(\Theta)= \Theta$ without loss of generality).  
The set $(\Pi,\hat{\theta})$ is to be called a quantum estimator or simply an estimator 
and is denoted as $\hat{\Pi}=(\Pi,\hat{\theta})$. 
Throughout our discussion, we restrict ourselves to find the best estimator satisfying the locally unbiased 
condition, that is for a given true value $\theta\in\Theta$, the estimator $\hat{\Pi}$ needs to satisfy 
the following condition for all $i,j=1,2,\dots,k$,
\be\label{lucond}
\sum_{x\in\cX} \hat{\theta}_i(x) \tr{\rho_{\theta}\Pi_x}=\theta_i,\  
\sum_{x\in\cX} \hat{\theta}_i(x)\tr{\del_j\rho_{\theta}\Pi_x}=\delta_{ij},
\ee
where $\del_i=\del/\del\theta_i$ is the partial derivative about $\theta_i$ and $\delta_{ij}$ is the Kronecker delta. 
We remind that this locally unbiased condition is much weaker than unbiased condition where 
one demands $\sum_{x\in\cX} \hat{\theta}_i(x) \tr{\rho_{\theta}\Pi_x}=\theta_i$ holds for all values of $\theta\in\Theta$. 

A problem of finding an optimal (locally unbiased) quantum estimator $\hat{\Pi}$ is to minimize the MSE matrix 
$V_{\theta}[\hat{\Pi}]$ for a given model. 
In contrast to a (classical) parameter estimation problem, a quantum problem, 
however, does not exhibit the general solution as a matrix inequality 
except for special cases. One tractable formulation of the problem is 
to minimize a weighted trace of the MSE matrix, which is a scalar 
quantity; $\mathrm{Tr}\{W V_\theta[\hat{\Pi}]\}$.  
Here a $k\times k$ positive matrix $W$ is called a weight matrix and 
can be chosen arbitrary. To distinguish traces for density matrices and 
MSE matrices, we use lower case letter for quantum state and 
upper case letter for the latter. 
Thus, our problem for a quantum parameter estimation problem is 
to find the precision bound which is defined as 
\begin{equation}  \label{opt}
 C_\theta[W]=\min_{\hat{\Pi}:\mathrm{l.u. at}\, \theta}  \mathrm{Tr}\{W V_\theta[\hat{\Pi}]\},
\end{equation}
where $\mathrm{l.u. at}\, \theta$ indicates the optimization is carried under the locally unbiased condition 
\eqref{lucond} and the optimal quantum estimator is denoted as $\hat{\Pi}_{opt}[W]$. 
%:=\mathrm{arg}\hspace{0mm}\smash{\min_{\hspace{-4mm}\hat{\Pi}:\mathrm{l.u. at}\, \theta}}   \mathrm{Tr}\{W V_\theta[\hat{\Pi}]\}$. 

As in (classical) estimation problems, we are given an $n$ copy of quantum states and 
is mathematically represented by a tensor product as $\rho_\theta^{\otimes n}=\bigotimes_{i=1}^n\rho_\theta$. 
This is analogous situation to identically and independently distributed (i.i.d.) scenario in probability theory 
and the state $\rho_\theta^{\otimes n}$ is referred to as an i.i.d. quantum state. 
Upon estimating a parameter $\theta$ for a given i.i.d. states, a significant difference arises for the quantum case. 
A quantum statistician can choose different strategies: One is to 
perform a POVM written as a tensor product $\Pi^{(n)}=\{ \Pi^{(1)}_x\otimes\Pi^{(2)}_x\otimes\dots \Pi^{(n)}_x\}_{x\in\cX} $, 
and the other is a general POVM on the joint Hilbert space $\cH^{\otimes n} $ which cannot be expressed as a tensor product. 
The former is called a separable measurement, and the latter is collective measurement in literature. 
It is known that collective measurements are more powerful than separable ones in general. 
In the following, we consider a separable measurement mainly and collective measurement scheme 
will be discussed in Section \ref{sec4}. 

One way to see why a quantum estimation problem is non-trivial is as follows. 
For (classical) estimation problems to estimate the probability distribution $p_{\theta}$, 
the fundamental precision bound for the MSE 
is given by the Cram\'er-Rao (CR) inequality 
which states that for any locally unbiased estimator the MSE matrix is bounded as 
\be
V_{\theta}[\hat{\theta}]\ge (J_\theta[p_\theta])^{-1}.  
\ee
In this inequality, $J_\theta[p_\theta]$ denotes the Fisher information matrix 
for a given probability distribution $p_\theta$ whose $(i,j)$ component is defined by 
\be \label{clfisher}
J_{\theta,ij}:=\sum_{x\in\cX} p_{\theta}(x) \del_i \ell_{\theta}(x) \del_j \ell_{ \theta}(x),
\ee
with $\del_i \ell_{\theta}(x)=\del_i\log p_{\theta}(x)$ the $i$th logarithmic derivative. 
This bound can be achieved asymptotically, for example, by the maximum likelihood estimator. 
For the quantum case, let us fix a measurement $\Pi$ on $\rho_\theta$ then the best 
estimator $\hat{\theta}$ should be given by the above CR bound as 
$V_{\theta}[\hat{\Pi}]\ge (J_\theta[\Pi])^{-1}$. Here the Fisher information 
matrix is calculated according to the probability distribution: $p_\theta(x)=\tr{\rho_{\theta}\Pi_x}$ 
and solely determined by the choice of a POVM.  
We remind ourselves that partial differentiations $\del_i$ must act only on the state 
in the probability distribution $\tr{\rho_{\theta}\Pi_x}$. 
To find the optimal estimator for a given quantum model is then reduced to 
minimize the inverse of Fisher information matrix $(J_\theta[ \Pi])^{-1}$ over all possible POVMs. 
This problem is rather difficult 
simply because of an optimization of non-scalar quantity over matrix spaces with certain constraints. 
Thus, the strategy to minimize the weighted trace of the inverse of Fisher information matrix 
is another view into quantum parameter estimation problem. Let us call 
\be\label{mibound}
C_\theta^{\mathrm{MI}}[W]:=\min_{\Pi: \mathrm{POVM}}\Tr{W (J_\theta[\Pi])^{-1}}, 
\ee
the most informative precision allowed by quantum mechanics. It is known that 
$ C_\theta[W]=C_\theta^{\mathrm{MI}}[W]$ holds in general \cite{nagaoka89}, 
and Fisher information plays an important role even in quantum parameter estimation theory.  

To define quantum version of logarithmic derivatives and Fisher information, 
we first introduce an inner product for any linear operators and 
then define quantum Fisher information based on the inner product. 
It happens that there is no unique way to define an inner product for quantum cases, 
meaning that we have many different quantum versions of Fisher information. 
In the following, we use two kinds of quantum Fisher information based on 
symmetric logarithmic derivative (SLD) and right logarithmic derivative (RLD) operators. 
For a given quantum state $\rho_{\theta}$ and any (bounded) linear operators $X,Y$ on $\cH$, 
we define symmetric and right inner product by 
\begin{align} \nonumber
\sldin{X}{Y}&:=\tr{\rho_{\theta}(YX^\dagger+X^\dagger Y)}, \\
\rldin{X}{Y}&:=\tr{\rho_{\theta}YX^\dagger}, 
\end{align}
respectively. 
SLD operators $L_i$ and RLD  operators $\tilde{L}_i$ are formally defined by the solutions to the operator equations:
\begin{align}\nonumber
\del_i\rho_{\theta}&=\frac12 (\rho_{\theta}L_{\theta,i}+L_{\theta,i}\rho_{\theta}), \\
\del_i\rho_{\theta}&=\rho_{\theta}\tilde{L}_{\theta,i}. 
\end{align}
It is not difficult to see that the SLD operators are hermite, whereas 
RLD operators are not in general. The SLD Fisher information 
matrix is defined by 
\begin{align} 
G_{\theta}&:= \left[ g_{\theta, ij}\right]_{i,j\in\{1,\dots,k\}} \\ \nonumber
g_{\theta, ij}&:=\sldin{L_{\theta,i}}{L_{\theta,j}}=\tr{\rho_{\theta}\frac12 \big(L_{\theta,i}L_{\theta,j}+L_{\theta,j}L_{\theta,i}  \big)}, 
\end{align}
and the RLD Fisher information is 
\begin{align} \nonumber
\tilde{G}_{\theta}&:= \left[ \tilde{g}_{\theta, ij} \right]_{i,j\in\{1,\dots,k\}}, \\
\tilde{g}_{\theta, ij}&:=\rldin{\tilde{L}_{\theta,i}}{\tilde{L}_{\theta,j}}=\tr{\rho_{\theta}\tilde{L}_{\theta,j}\tilde{L}_{\theta,i}^\dagger}. 
\end{align}

The quantum versions of CR inequality state that for any locally unbiased estimators its MSE matrix satisfies 
\begin{align}\nonumber
V_{\theta}[\hat{\Pi}]&\ge G_\theta^{-1},   \\
V_{\theta}[\hat{\Pi}]&\ge \tilde{G}_\theta^{-1}. 
\end{align}
These are referred to as the SLD CR inequality and RLD CR inequality, respectively.  
For notational convenience, the $(i,j)$ component of the inverse of 
the SLD Fisher information is denoted as $g_{\theta}^{ij}$, i.e., $G_\theta^{-1}=[g_\theta^{ij}]_{i,j\in\{1,\dots,k\}}$. 
Unlike the classical CR bound, there is no estimator $\hat{\Pi}$ in general attaining the 
equalities in the above inequalities. 
Combining the above considerations, one can show that for any POVMs 
the following relation holds: 
\be\label{qcmonotone}
V_{\theta}[\hat{\Pi}]\ge (J_{\theta}[\Pi])^{-1}\ge G_\theta^{-1},
\ee
and similarly for the RLD Fisher information. This inequality again emphasizes 
importance of Fisher information since the true bound lies in-between 
$J_\theta[\Pi]$ and $G_\theta$.  

Before closing this subsection, we have several remarks regarding quantum Fisher information. 
First, quantum Fisher information should be used as a collective noun rather than a proper noun 
since there are many quantum versions of Fisher information in general.  

Second, among existing many quantum Fisher information, SLD and RLD Fisher information 
stand as special ones \cite{petz}. The SLD Fisher metric is known as the minimum operator-monotone metric 
whereas the RLD Fisher metric is the maximum one. This is a well-known result, but this does not 
imply the matrix inequality $\tilde{G}_\theta\ge G_\theta$ in general. That is, 
there is no ordering between $\tilde{G}_\theta$ and $G_\theta$ in general. 
The valid relation holds for real part of the RLD Fisher information and SLD Fisher information:  
\be\label{rerldsld}
\mathrm{Re}\, \tilde{G}_\theta\ge G_\theta, 
\ee
for any quantum parametric models. 

Third, the RLD Fisher information always dominates the SLD Fisher information 
when the number of parameters is equal to one. In this case, the SLD Fisher 
information is attainable by the projection measurement with respect to 
the spectral decomposition of the SLD operator, and RLD Fisher information 
does not provide important information as long as state estimation is concerned.  

Fourth, as in classical statistics, we assume some regularity condition for quantum 
parametric models to define quantum versions of Fisher information.  
Besides mathematical technical assumptions, the rank of a state is important. 
When the state is not full-rank, it is known that SLD operators 
and SLD Fisher information cannot be defined uniquely. 
However, modification of the inner products by taking an equivalent class 
provides a well-defined and unique SLD Fisher information \cite{fn95}. 

Last, quantum Fisher information is proper information quantity and satisfy 
important properties. To list a few: i) They are semi-definite positive matrix. 
ii) They do not increase when a quantum operation (completely positive map) is 
applied to the state. iii) They are convex with respect to quantum states. 
iv) They are additive for product states. 

\subsection{SLD CR , RLD CR, and Holevo bounds}
Within our formulation of the problem, there are several bounds for the weighted trace 
of the MSE matrix \eqref{opt}. The first one is the SLD CR bound defined by
\be\label{sldcr}
C_\theta^S[W]:=\Tr{WG_\theta^{-1}},  
\ee
and this leads to the bound for any locally unbiased estimators as
\be
\Tr{WV_\theta[\hat{\Pi}]}\ge C_\theta^S.
\ee

The second one utilizes the RLD Fisher information and 
the following relation; 
\[
V\ge X\ \Rightarrow\  \Tr{WV}\ge\Tr{W\mathrm{Re}X}+\Trabs{W \mathrm{Im}X},
\]
for a positive matrix $W$, real symmetric matrix $V$, and Hermite matrix $X$. 
Here, TrAbs$X$ denotes the trace of absolute values of eigenvalues of the matrix $X$, 
i.e., TrAbs$X=\sum_i|\lambda_i| $ with $X=\sum_i \lambda_i \ket{i}\bra{i}$ an 
eigenvalue decomposition of $X$. Since the RLD Fisher information is 
complex-valued in general, the above inequality gives the RLD CR bound: 
\be
C_\theta^R[W] :=\Tr{W\mathrm{Re}\,\tilde{G}_\theta^{-1}}+\mathrm{TrAbs}\{W\mathrm{Im}\,\tilde{G}_\theta^{-1}\} . 
\ee

The bound for quantum model which unifies the above bounds is due to Holevo 
and it is referred to as the Holevo bound \cite{holevo}. 
Denote a $k$ array of Hermite operators on $\cH$ by 
\[
\vec{X}:=(X^1,X^2,\dots, X^k),\quad (X^\ell)^\dagger=X^\ell\ (\ell=1,2,\dots,k),
\]
and define the set of $\vec{X}$ by
\be \label{Xset}
\cX_\theta:=\{\vec{X}\,|\, \forall i\,\tr{\rho_\theta X^i}=0,\ \forall i,j\,\tr{\del_i\rho_\theta X^j}=\delta_{ij} \}.
\ee
The holevo function for quantum estimation is defined by
\be
h_\theta[\vec{X}|W]:=\Tr{W\mathrm{Re}\,Z_\theta[\vec{X}]}+\Trabs {W\mathrm{Im}\,Z_\theta[\vec{X}] }, 
\ee
where the $k\times k$ matrix $Z_\theta[\vec{X}]$ is
\be
Z_\theta[\vec{X}]:= [ \rldin{X^i}{X^j}]_{i,j\in\{1,\dots,k\}}. 
\ee
The Holevo bound is defined through the following optimization:
\be
C_\theta^H[W]:=\min_{\vec{X}\in\cX_\theta}h_\theta[\vec{X}|W]. 
\ee
Importantly, any locally unbiased estimators is bounded by the Holevo bound as 
\be
\Tr{WV_\theta[\hat{\Pi}]}\ge C_\theta^H[W], 
\ee
The Holevo bound can be attained asymptotically by 
an asymptotically unbiased estimator with a collective POVM in the following sense \cite{HM08,GK06,KG09,YFG13}. 
Consider $n$th i.i.d. quantum state $\rho_\theta^n=\rho_\theta^{\otimes n}$ for 
a given model and 
let $\hat{\Pi}^n$ be a sequence of estimators for the model $\cM^n_Q=\{\rho^n_\theta\,|\,\theta\in\Theta  \}$. 
An estimator is called asymptotically unbiased if the locally unbiased condition \eqref{lucond} 
hold for all values of $\theta$ in the $n\to\infty$ limit. Let us denote the MSE for the $n$th extension 
model as $V^n_\theta$, then, the Holevo bound has the operational meaning: 
\begin{multline}
C_\theta^H[W]=\inf\big\{ \limsup_{n\to\infty}\,n\Tr{W V^n_\theta[\hat{\Pi}^n]}\,\\
\big|\,\hat{\Pi}^n \mbox{ is asymptotically unbiased} \big\}. 
\end{multline}
%\be
%C_\theta^H[W]=\inf\big\{ \limsup_{n\to\infty}\,n\Tr{W V^n_\theta[\hat{\Pi}^n]}\,\big|\,\hat{\Pi}^n \mbox{is asymptotically unbiased.} \big\}. 
%\ee 
That is the optimal MSE behaves as $\Tr{W V^n_\theta[\hat{\Pi}^n]}\simeq C_\theta^H[W]/n$ for very large $n$ 
by performing the optimal sequence of collective POVMs. 
In this sense, the Holevo bound is considered as the {\it ultimate bound} in quantum 
parameter estimation problem. 

Several remarks are listed regarding relations among the SLD CR, RLD CR, and Holevo bounds. 
First, there are no ordering in general between the SLD CR bound and the RLD CR bound, 
despite the fact \eqref{rerldsld}. When the inverse of RLD Fisher information matrix 
has no imaginary entries, then \eqref{rerldsld} gives $C_\theta^S[W]\le C_\theta^R[W]$. 
This indicates importance of imaginary part of the RLD Fisher information. 
Second, the Holevo bound is always superior both to SLD CR and RLD CR bounds, i.e., 
\be
C_\theta^H[W]\ge C_\theta^S[W]\  \mathrm{and}\ C_\theta^H[W]\ge C_\theta^R[W].  
\ee
Third, when the number of parameters equal to one, the Holevo bound 
is identical to the SLD CR bound. 
Thus, collective measurements do not help to improve the accuracy of estimation. 
Fourth, when a model is so called D-invariant \cite{holevo, HM08}, 
the Holevo bound and the RLD CR bound coincide. In this case, 
the Holevo bound can be expressed as
\be
C_\theta^H[W]=h_\theta[\vec{L}_\theta|W]\quad (\mbox{for D-invariant model}), 
\ee
where $\vec{L}_\theta=(L^1_\theta,L^2_\theta,\dots,L^k_\theta)$ is the cotangent 
vector of SLD operators, i.e., $L^j_\theta=\sum_{i=1}^k (G_\theta^{-1})_{ij}L_{\theta,i}$ ($j=1,2,\dots,k)$. 
Thus, a D-invariant model possesses nice structure as a statistical model, 
and this condition is satisfied, for example, when the set of SLD operators together with the identity operator 
span the whole Hermite operators, i.e., $\mathrm{span}_{\bbr}\{ L_{\theta,1},L_{\theta,2},\dots,L_{\theta,k},I\}=\lofhh$ holds.

\subsection{Nagaoka and Hayashi-Gill-Massar bounds}
For two-dimensional quantum system, the quantum estimation problem 
is completely solved and the attainable bound can be calculated for an arbitrary quantum statistical model. 
This problem was solved for two-parameter case affirmatively by Nagaoka in the 80s \cite{nagaoka89}. 
Hayashi solved the case for three-parameter by utilizing the infinite dimensional 
linear programming method \cite{hayashi97}. Gill and Massar independently 
solved the same problem by different manner \cite{GM00}. 
We recommend a compact proof by Yamagata \cite{yamagata}. 
In the rest of paper, we call the bound for two-parameter case as 
the Nagaoka bound and the one for three-parameter case 
as the Hayashi-Gill-Massar (HGM) bound 
for the sake of convenience although the latter includes the former as a special case. 

Consider a complex two-dimensional Hilbert space $\bbc^2$ and a 
quantum parametric model on it.  For a given weight, 
the Nagaoka and the HGM bound for the weighted trace of MSE is given by
\be \label{n1bd}
\min_{\hat{\Pi}:\mathrm{l.u. at}\, \theta}  \mathrm{Tr}\{W V_\theta[\hat{\Pi}]\}
=\left(F(G_{\theta}^{-1},W)\right)^2=: C^{HGM}_{\theta}[W],
\ee
where $F(A,B)=\Tr{\sqrt{\sqrt{A}B\sqrt{A}}}$ denotes a fidelity 
between two semi-definite positive operators $A,B$. 
Nagaoka proved that this bound is more informative that the other bounds, i.e., 
$C^{HGM}_\theta[W]\ge C^H_\theta[W]$ holds \cite{nagaoka89}. 
The achievability of the HGM bound is known as the necessary and sufficient condition for POVMs, 
which states that the minimum is attained if and only if a POVM satisfies the condition \cite{yamagata}:
\be \label{optcond}
J_{\theta}[\Pi_{opt}]=\frac{ \sqrt{G_{\theta}} \sqrt{F_{\theta}} \sqrt{G_{\theta}}}{\Tr{\sqrt{F_{\theta}}}}. 
\ee

One way to compute the fidelity between $A$ and $B$ is to 
calculate the eigenvalues of the hermite operator $\sqrt{A}B\sqrt{A}$, 
and to compute the sum of square root of all eigenvalues.  
To proceed further we introduce the following $k\times k$ real symmetric matrix 
and assume its diagonalized form as  
\be\label{diagrep}
F_{\theta}=\sqrt{G_{\theta}^{-1}}W\sqrt{ G_{\theta}^{-1}}  
=U_{\theta}\Lambda_{\theta}U_{\theta}^{-1}, 
\ee
where $U_{\theta}$ is real orthogonal matrix and $\Lambda_{\theta}$ is a diagonalized 
matrix whose elements are the eigenvalues of $F_{\theta}$ given as 
$\Lambda_{\theta}=\mathrm{diag}(\lambda_1,\lambda_2,\dots,\lambda_k)$. 
With these notations, the HGM bound $\eqref{n1bd}$ is expressed as 
\be\label{n1bd2}
C^{HGM}_\theta[W]
=\sum_{i=1}^{k}\lambda_i+\sum_{i\neq j}\sqrt{\lambda_i\lambda_j}. 
\ee
Using the fact that the sum of eigenvalues of hermite matrix is equal to its trace, 
the first term in the right hand side of Eq.~\eqref{n1bd2} is written as 
$\sum_{i=1}^{k}\lambda_i=\Tr{WG_{\theta}^{-1}}$. 
This term corresponds to the SLD CR bound and 
noting the eigenvalues of the matrix $F_{\theta}$ are positive in general, 
we see that the HGM bound is strictly larger than the SLD CR bound for any weight matrix $W$, i.e., 
\be
C^{HGM}_\theta[W]>\Tr{WG_{\theta}^{-1}}. 
\ee
This shows that the SLD CR bound cannot be 
attained for a generic qubit problem. Two well-known exceptions are: 
one-parameter case ($k=1$) and the case where all SLD operators commute with each other. 

An optimal POVM was explicitly constructed by Nagaoka \cite{nagaoka91,FN99}, 
and its general form is given as follows \cite{yamagata}. 
Let $\hat{L}_{\theta,i}$ ($i=1,2,\dots,k$) be a linear combination of 
SLD operators defined by 
\be \label{sldhat}
\hat{L}_{\theta,i}:= \sum_{j=1}^k \left[ U_{\theta}^{-1}\sqrt{ G_{\theta}^{-1}}\right]_{ij}L_{\theta,j},
\ee
where $U_{\theta}$ is the matrix diagonalizing $F_{\theta}$ in Eq.~\eqref{diagrep}, 
and $G_{\theta}$ is the SLD Fisher information matrix. 
Let ${\Pi}^{(i)}$ ($i=1,2,\dots,k$) be the projection measurement, 
or projection valued measure (PVM), about the observable $\hat{L}_{\theta,i}$, 
then the optimal POVM which attains the HGM bound is 
to perform the PVMs ${\Pi}^{(i)}$ with a corresponding probability: 
\be
p_i=\frac{\sqrt{\lambda_i}}{\sum_{j=1}^k\sqrt{\lambda_j}}.
\ee
Explicitly, writing ${\Pi}^{(i)}=\{\Pi_{i\pm}\}$ for binary outcome PVMs, 
the optimal POVM consists of $2k$ elements and is given by
\be
\Pi_{opt}=\left\{p_1\Pi_{1\pm},\dots,p_k\Pi_{k\pm}  \right\}.
\ee
The optimal estimator $\hat{\theta_i}$($i=1,\dots,k$) is to assign the following estimate 
upon the measurement outcomes:
\be \label{optest}
\hat{\theta}_i(x)=\theta_i+\sum_{j=1}^{k}\left(J_{ \theta}^{-1}\right)_{ij}\del_j \log p_{\theta}(x), 
\ee
with $x\in\cX=\{1\pm,\dots,k\pm\}$ and $p_{\theta}(x)=\tr{\rho_{\theta}\Pi_x }$. 
We remark that there are other forms of optimal POVMs known in literature \cite{GM00,hayashi97}

This optimal estimator $(\Pi,\hat{\theta})$ explicitly depends on 
the true value of the unknown parameter $\theta$. This might be considered 
as self-contradiction in the formalism. Indeed, some authors claim that 
finding unbiased estimator is rather purely of mathematical interest and is of no use. 
Here we mention that the formalism	based on locally unbiased estimators 
needs an additional ingredient when applying to real problem.  
It was first proposed by Nagaoka that one should perform the above 
optimal estimation adaptively, namely, when one starts with unknown value 
of parameters and then successively update the values according to 
measurement results \cite{nagaoka89-2}. 
The mathematical rigorous proofs for strong consistency and asymptotic efficiency 
of adaptive estimation are due to Fujiwara \cite{fujiwara06}. 
We take these mathematical justifications for granted to look for locally unbiased estimators. 

There is also an alternative way to achieve the bound 
obtained for the locally unbiased estimators 
by using two-step estimation strategy \cite{HM98,BNG00}. 
In this method, one take a few fraction of $n$ copies, say $\sqrt{n}$, to 
estimate the value of $\theta$ and then perform the optimal locally unbiased estimator 
for the remaining $n-\sqrt{n}$ copies. 
Finally, we remark that Yamagata shows that 
the adaptive estimation method works more efficiently than the standard quantum tomographic scheme 
in qubit system \cite{yamagata}. 
An adaptive estimation scheme for one-parameter case was experimentally demonstrated in Ref.~\cite{oioyift12}. 

\section{Estimation of qubit states in presence of unknown phase}\label{sec3}
The quantum parametric model under consideration is given 
by the family of quantum states on the two-dimensional Hilbert space $\bbc^2$, 
i.e, qubit states:  
\be
\rho_{\theta}=\frac{1}{2}\left(\begin{array}{cc}1+\theta_2 & \theta_1 \Exp{-\I\theta_3} \\[.5ex] 
\theta_1 \Exp{\I\theta_3} & 1-\theta_2\end{array}\right) \in \sofc2, 
\ee
where $\theta=(\theta_1,\theta_2,\theta_3)$ satisfy $\theta_1^2+\theta_2^2<1$ and $\theta_3\in[0,2\pi)$ 
and we exclude the point $\theta_1=0$ for the sake of mathematical convenience. 
This condition guarantees that the state is full-rank for all values of $\theta\in\Theta$. 
It is useful to go from matrix representation of a state to the three dimensional vector representation, 
so called the Bloch vector representation. 
This is given by a one-to-one mapping as
%\be
$\sofh\ni\rho \longmapsto \v{s}=\tr{\v\sigma \rho}\in \mathbb{R}^3$,  
%\ee
where $\v \sigma=(\sigma_1,\sigma_2,\sigma_3)^T$ denotes 
the set of usual Pauli spin operators. 
Requirement of unit trace and positivity imposes on the vector 
such that the length of this vector is less or equal to one. 
The space of all possible Bloch vectors is defined as 
surface and interior of unit sphere: 
$\cB =\left\{b\in\mathbb{R}^3\,\big|\, |b|\le1  \right\}$. 
The inverse mapping from a given Bloch vector $\v s$ to matrix representation is 
%\be
$\cB\ni \v{s}\longmapsto \rho=(\sigma_0+\v{s}\cdot\v{\sigma}  )/2\in\sofh$, 
%\ee
with $\sigma_0$ the identity $2\times2$ matrix and 
$\v{s}\cdot\v{\sigma}=\sum_{i=1,2,3}s_i\sigma_i$. 
Thus, the Bloch vector representation of  our quantum model is  
\be
\v{s}_{\theta}=(\theta_1\cos\theta_3,\,\theta_1\sin\theta_3,\,\theta_2)^T. 
\ee
For later convenience, we define the standard inner product and 
the outer product for three dimensional vectors by 
\begin{align*}
\vecin{\v a}{\v b}&=\sum_{i=1,2,3}a_ib_i,\\
\ket{\v a}\bra{\v b}&=[a_ib_j]_{i,j\in\{1,2,3\}},
\end{align*}
respectively. The outer product is $3\times3$ matrix whose action onto a vector 
$\v c\in\bbr^3$ is $\ket{\v a}\bra{\v b} {\v c}=\vecin{\v b}{\v c}{\v a}$.

In our problem, the phase parameter $\theta_3$ is of no interest, which is called a {\it nuisance parameter}, 
and we wish to discuss how well we can estimate this parametric model in presence 
of the nuisance parameter $\theta_3$. In the following, we first solve the problem 
when $\theta_3$ is known (no nuisance parameter) and then to solve the case with unknown phase parameter. 
In our analysis, we introduce an important concept, {\it mean square error region} defined as follows \cite{comment}. 
Given a quantum parametric model $\cM_Q$ 
and the bound for the MSE matrix $ \mathrm{Tr}\{W V_\theta[\hat{\Pi}]\}\ge C_\theta[W]$, 
we define the set of all possible MSE matrices allowed by locally unbiased estimators:  
\be
D_{\mathrm{l.u.}}=\{V\in M_h\,\big|\, V=V_{\theta}[\hat{\Pi}] ; \hat{\Pi} \mbox{ is locally unbiased at }\theta  \},
\ee
and the set of positive matrices allowed by the given bound: 
\be
D_C:= \{V\in M_{h}\, |\, \Tr{WV}\ge C_\theta[W],\  \forall W>0\}.
\ee
Here $M_{h}$ denotes the set of all $2\times2$ symmetric matrices, i.e., 
$M_{h}=\left\{ V\in\mathbb{R}^{2\times 2}\,\Big|\, V^T=V\right\}$.  
It is not difficult to show that two sets are equivalent, i.e., $D_{\mathrm{l.u.}}=D_C$, 
if the bound is achievable. We call $D_{\mathrm{l.u.}}$ a MSE region and 
we shall analyze $D_C$ in the following based on the Nagaoka and the HGM bound.

\subsection{No nuisance parameter case}
In this subsection, we assume that the phase parameter $\theta_3$ is known with infinite precision. 
The number of parameters to be estimated is two and the straightforward calculation shows 
the inverse of the SLD Fisher information is given by
\be \label{sldF2}
G_{\theta}^{-1}
=\left(\begin{array}{cc}1-\theta_1^2 & -\theta_1\theta_2 \\[.5ex] 
 -\theta_1\theta_2  & 1-\theta_2^2\end{array}\right). 
\ee
With this simple structure of SLD Fisher information, we have 
$ \Tr{G_{\theta}^{-1}}-1=\det G_{\theta}^{-1}=1-s_{\theta}^2$ where $s_{\theta}=(\theta_1^2+\theta_2^2)^{1/2}$ 
denotes the length of the Bloch vector. 
Eigenvalues of the $2\times 2$ matrix defined for the Nagaoka bound \eqref{diagrep} are given by 
\begin{align} \label{lambda} 
&\lambda_{1,2}=\frac12\left(t_{\theta}\pm\sqrt{\Delta_{\theta}} \right),\\\nonumber
&t_{\theta}=\Tr{WG_{\theta}^{-1}},\ \Delta_{\theta}= t_{\theta}^2-4\det (WG_{\theta}^{-1}). 
\end{align}
The Nagaoka bound $C_\theta^N[W]$ is thus written as 
\be \label{nhgm2}
C^{N}_\theta[W]=\Tr{WG_{\theta}^{-1}}+2\sqrt{\det WG_{\theta}^{-1}}.
\ee
An optimal POVM attaining this bound is given as follows.  
Consider rank-1 projectors:
\be \label{optPVM2}
P_{\theta,1(2)}=E_\theta\,\frac{{W}-{\lambda}_{2(1)}G_{\theta}}{\Tr{{W}}-{\lambda}_{2(1)}\Tr{G_{\theta}}}\,E_\theta^T, 
\ee
where $E_\theta$ is a $3\times2$ real matrix,
\be\label{ematrix}
E_\theta:= \left(\begin{array}{cc}\cos\theta_3 & 0 \\ \sin\theta_3 & 0 \\0 & 1\end{array}\right), 
\ee 
and write them as $P_{\theta,i}=\ket{\v{n}_i}\bra{\v{n}_i}$ with $\v{n}_i$ unit vectors. 
We remark that $\v{n}_1$ and $\v{n}_2$ are not orthogonal in general. 
An optimal POVM is then written as 
\begin{align} \label{optPOVM2}
\Pi_{opt}&=\left\{p_1\Pi_{1+},\, p_1\Pi_{1-},\,p_2\Pi_{2+},\, p_2\Pi_{2-}  \right\},\\ \nonumber
\Pi_{i\pm}&=\frac{1}{2}\left(\sigma_0\pm\v{n}_i\cdot \v\sigma  \right) ,\ 
p_{1,2}=\frac12(1\pm \cos 2 q_\theta), 
%p=\Big[\frac{t_{\theta}-2\sqrt{\det WG_{\theta}^{-1}}}{t_{\theta}+2\sqrt{\det WG_{\theta}^{-1}}}\Big]^{1/2}. 
\end{align}
where $q_\theta= \arctan [(1-\Delta_\theta/t_\theta^2)^{1/4}]$. 
With this optimal POVM and the estimator \eqref{optest}, the value of MSE matrix 
is given by 
\be \nonumber
V_{\theta}[\hat{\Pi}_{opt}]=J_{\theta}[\Pi_{opt}]\,^{-1}=G_{\theta}^{-1}+\sqrt{\det WG_{\theta}^{-1}}\, {W}^{-1}. 
\ee

As an example, let us consider the case for the identity wight matrix, 
which corresponds to estimating two parameters with equal footing. 
In this case, the optimal POVM takes a rather simple form:
\begin{align*}
\Pi_{1\pm}&=\frac12 (\sigma_0\pm \frac{\v{s}_{\theta}}{s_\theta}\cdot \v\sigma)\ \mbox{with}\  
p_1=\frac{\sqrt{1-s_\theta^2}}{1+\sqrt{1-s_\theta^2}},\\
\Pi_{2\pm}&=\frac12 (\sigma_0\pm \frac{\v{s}^{\bot}_{\theta}}{s_\theta}\cdot \v\sigma)\ \mbox{with}\  
p_2=\frac{1}{1+\sqrt{1-s_\theta^2}}, 
\end{align*}
and the estimator $\hat{\theta}(\pm)=\left(\hat{\theta}_1(\pm),\hat{\theta}_2(\pm)\right)$:
\begin{align*}
\left(\hat{\theta}_1(\pm),\hat{\theta}_2(\pm)\right)&=\Big[1\pm\frac{1}{p_1}\frac{1-s_\theta}{s_\theta}\Big]({\theta}_1,{\theta}_2)\quad \mathrm{for}\ \Pi_1,\\
\left(\hat{\theta}_1(\pm),\hat{\theta}_2(\pm)\right)&=({\theta}_1,{\theta}_2)\pm\frac{1}{p_2s_\theta}(-\theta_2,\theta_1)\ \mathrm{for}\ \Pi_2. 
\end{align*}
Since the state under estimation is given by $\rho_\theta=(\sigma_0+\v{s}_\theta\cdot\v{\sigma}  )/2$, 
the first PVM $\Pi_1$ suggests to measure along the same direction as the unknown state. 
The second PVM, however, suggests us to measure along the perpendicular direction 
$\v{s}^{\bot}_{\theta}=(\theta_2\cos\theta_3,\,\theta_2\sin\theta_3,\,-\theta_1)^T$. 
This seems rather counter intuitive since the probability distribution upon 
the measurement $\Pi_2$ on $\rho_\theta$ is $1/2$, i.e., completely random outcomes, 
and hence this does not provide us any useful information to estimate the value $\theta$. 
To understand this optimal estimator, we note that both PVMs and estimators 
do depend on the true values of parameters $\theta$ and we can only 
attain this optimal quantum estimator by adaptively in $n\to\infty$ limit. 
As emphasized before, this is one formulation of quantum parameter estimation 
within locally unbiased estimators.  

To characterize the MSE region obtained from the Nagaoka bound 
for this problem, we note the following fundamental lemma: 
\begin{lemma}\label{lemma1}
Let $c$ be a positive constant, the following two sets are equivalent. 
\begin{align}\nonumber
D_1&=\{V\in M_{h}\,|\,\Tr{XV}\ge 2c\sqrt{\det X},\ \forall X>0 \}, \\ \nonumber
D_2&=\{V\in M_{h}\,|\,\det V \ge c^2,\  V>0 \}.
\end{align} 
\end{lemma}
\begin{proof}
This lemma can be shown as follows. 
From $\Tr{XV}\ge 2c\sqrt{\det X}>0$ for all $X>0$, we have $V>0$. 
Change the positive matrix as $X\to V^{-1/2} X V^{-1/2}> 0$, 
we have , $\Tr{X}\ge 2c\sqrt{\det XV^{-1}}\Leftrightarrow \det{V}\ge 2c \sqrt{\det{X}}/\Tr{X}$. 
Note for $2\times2$ matrix $X$ 
a functional $f(X):=\sqrt{\det X}/\Tr{X}$ has the maximum and $\forall X>0, 1/2\ge f(X)>0$ holds. 
With this we can show $D_1\subset D_2$. This argument can be reversed to show the converse inclusion. 
$\square$
\end{proof}

With this lemma, we state our first result: 
\begin{proposition}
The following sets are all equivalent. 
\begin{align*}
%\hspace{-0mm}D_{\ell.u.}&=\left\{V_{\theta}[\hat{\Pi}] \,\big|\, \hat{\Pi} \mbox{ is locally unbiased at }\theta  \right\},\\
D_{N}&=\{V\in M_{h}\,|\,   \Tr{W V}\ge C^{N}_\theta[W],\  \forall W>0 \},\\
D_{GM}&=\{V\in M_{h}\ \,|\, \Tr{G_\theta^{-1} V^{-1}}\le 1,\ V>G_{\theta}^{-1} \}, \\
D&=\{V\in M_{h}\ \,|\, \det(V-G_{\theta}^{-1})\ge \det G_{\theta}^{-1},\ V>G_{\theta}^{-1}\}.  
%&\hspace{4cm}  v_{ww} >g^{ww}, v_{pp}>g^{pp} \Big\}  
\end{align*}
\end{proposition}
\begin{proof}
Equivalence between $D_{N}$ and $D$ follows from lemma \ref{lemma1}.  
The other relation $D=D_{GM}$ follows from a direct calculation which 
shows $\Tr{AV^{-1}}\le 1\Leftrightarrow \det (V-A)\ge \det A$ for all $A\in M_h$ and $V>0$. $\square$
\end{proof}

From the expression of $D$, we see that in general there is trade-off relation 
between errors in $\theta_1$ and $\theta_2$. 
We note that a similar trade-off relation was obtained in Ref.~\cite{wsu11} 
for any two observables in any finite dimensional quantum systems. However, 
their result heavily depends on the choice of parametrization for quantum states, 
and they only consider diagonal elements of the MSE matrix. We 
emphasize that all entries in the MSE matrix are important and the 
most general trade-off is 
\be
\det(V_\theta[\hat{\Pi}]-G_{\theta}^{-1})\ge \det G_{\theta}^{-1},
\ee
whereas SLD CR bound gives $\det(V_\theta[\hat{\Pi}]-G_{\theta}^{-1})>0$ 
and $\Tr{V_\theta[\hat{\Pi}]-G_{\theta}^{-1}}>0$. 

\subsection{Nuisance parameter case}
In this section, we treat the phase parameter $\theta_3$ as unknown and 
discuss how well we can estimate the parameter $\theta_1,\theta_2$. 
Let us first briefly recall for classical parameter estimation theory with nuisance parameters \cite{lc,bnc}. 
Consider a probability distribution 
$p_{\theta}(x)$ on $\cX$, $\theta=(\theta_1,\theta_2,\dots,\theta_k)$, where 
$\theta_{I}=(\theta_1,\theta_2,\dots,\theta_p)$ (parameters of interest) and 
$\theta_{N}=(\theta_{p+1},\theta_{p+2},\dots,\theta_k)$ (parameters of no interest, nuisance parameters). 
Let $J_{\theta}$ be the Fisher information matrix and consider block matrix representations as 
\be\nonumber
J_\theta=\left(\begin{array}{cc}J_{\theta_I\theta_I} & J_{\theta_I\theta_N} \\[0.1ex] 
J_{\theta_N\theta_I}& J_{\theta_N\theta_N}\end{array}\right), \ 
J_\theta^{-1}=\left(\begin{array}{cc}J^{\theta_I\theta_I} & J^{\theta_I\theta_N} \\[0.1ex] 
J^{\theta_N\theta_I}& J^{\theta_N\theta_N}\end{array}\right), 
\ee
in terms of the parameter grouping $\theta=(\theta_I,\theta_N)$. 
When $\theta_{N}$ are completely known, MSE for any unbiased estimators obeys 
\[
\ds V_{\theta}\ge ( J_{\theta_I\theta_I})^{-1}, 
\]
where all known values for $\theta_N$ are substituted to $p\times p$ matrix $J_{\theta_I\theta_I}$. 
When $\theta_{N}$ are not known, on the other hand, the MSE satisfies 
\[
\ds V_{\theta}\ge J^{\theta_I\theta_I}. 
\]
It is well-known that the following matrix inequality 
\begin{align}\nonumber
J^{\theta_I\theta_I}&=(J_{\theta_I\theta_I}-J_{\theta_I\theta_N}J_{\theta_N\theta_N}^{-1}J_{\theta_N\theta_I})^{-1}\\
&\ge ( J_{\theta_I\theta_I})^{-1}, \label{nuiineq}
\end{align}
holds where the equality holds if and only if the off-diagonal block matrix vanishes $J_{\theta_I\theta_N} =0$.  
In this case we say that two sets of parameters $\theta_I$ and $\theta_N$ are orthogonal with respect 
to Fisher information. These two CR inequalities show that 
when the MSE becomes in general worse in presence of nuisance parameters. 

We now consider our problem for quantum case. We can show that 
the primary parameter $\theta_I=(\theta_1,\theta_2)$ are orthogonal to the nuisance parameter $\theta_3$ 
with respect to the SLD Fisher information, and 
the inverse of SLD Fisher information matrix for three parameter case reads 
\be \label{sldF3}
G_{\theta}(3)^{-1}=\left(\begin{array}{cc}G_{\theta}^{-1} & \begin{array}{c}0 \\[-1mm] 0\end{array} \\[0ex]
0\ 0 & g_\theta^{33}\end{array}\right),
\ee
where $G_{\theta}^{-1}$ is same matrix given in \eqref{sldF2} and $g_\theta^{33}=1/\theta_1^2$. 
Clearly, $g^{33}_\theta$ diverges when $\theta_1=0$ simplely because 
we cannot have any information about $\theta_3$ at this point. Physically, 
this singularly is trivial since we cannot have any interference fringe at $\theta_1=0$. 
Thus, we justify the reason why we excluded the point $\theta_1=0$ in our model. 

This structure of the SLD Fisher information matrix might suggest that 
the bound \eqref{nhgm2} shown in the previous section holds even in presence of the nuisance parameter $\theta_3$. 
It is, however, not clear how to attain this bound without knowing the value of $\theta_3$. 
This is due to the fact that the optimal measurement \eqref{optPOVM2} explicitly depends on 
the unknown phase $\theta_3$, in particular the projectors \eqref{optPVM2}. 
To treat the effect of nuisance parameter in quantum case, we need 
to study the problem for estimating three parameters and to discuss 
trade-off between errors in $\theta_I$ and $\theta_N$. 

The HGM bound for generic three-parameter case can be written down as shown before. 
For the general $3\times3$ weight matrix, we have not gotten 
a simple expression for the HGM bound $C^{HGM}_\theta[W]$. 
To proceed our analysis, we write $3\times 3$ MSE and consider a special class of weight matrix as follows. 
\be\nonumber
V_{\theta}^{(3)}=\left(\begin{array}{cc}V_{\theta} & \begin{array}{c}v_{13} \\[-1mm] v_{23}\end{array} \\[0ex]
v_{31}\ v_{32} & v_{33}\end{array}\right),\  \label{bdform}
W(3)=\left(\begin{array}{cc}W& \begin{array}{c}0 \\[-1mm] 0\end{array} \\ 0\ 0 & w_3\end{array}\right),
\ee
where $V_{\theta}$ and $W$ are $2\times2$ matrices analyzed before. 
For this specific choice of the weight matrix, the HGM bound can 
be expressed in terms of $C^{N}_\theta[W] $ \eqref{nhgm2} as 
\begin{align} \nonumber
&\Tr{V_{\theta}^{(3)}W(3)}= \Tr{W V_\theta[\hat{\Pi}]}+w_3 v_{33} \ge C^{HGM}_\theta[W(3)],\\ 
& C^{HGM}_\theta[W(3)]=\Big(\sqrt{C^{N}_\theta[W]}+ \sqrt{w_3g_\theta^{33}}\Big)^2. \label{nhgm3}
 \end{align}
Let us denote the set of all symmetric and nonnegative positive $3\times 3$ matrices by 
$M_{h}^{(3)}$ and $M_{+}^{(3)}$, respectively, define the sets of positive matrices by
\begin{align*}
D_{HGM}&=\{V\in M_{h}^{(3)}\,|\,   \Tr{W V}\ge C^{HGM}_\theta[W],\  \forall W>0 \},\\
\tilde{D}_{HGM}&=\{V\in M_{+}^{(3)} \,|\,   \Tr{W(3) V}\!\ge\! C^{HGM}_\theta[W(3)],  \forall W(3)>0 \},\\
%D(3)&=\{V>0 \,|\, \det(V_2- \gamma_{\theta} G_{\theta}^{-1})\ge \det (\gamma_{\theta} G_{\theta}^{-1}), V_2>\gamma_{\theta} G_{\theta}^{-1},\ v_{33}>g_{\theta}^{33}\},   
D(3)&=\{V\in M_{+}^{(3)} \,|\, \det(V_2- \gamma_{\theta} G_{\theta}^{-1})\ge \det (\gamma_{\theta} G_{\theta}^{-1}),\\
&\hspace{3cm} V_2>\gamma_{\theta} G_{\theta}^{-1},\ v_{33}>g_{\theta}^{33}\},   
\end{align*}
where $\gamma_{\theta}$ is an important quantity defined by
\be \label{gamma1}
\gamma_{\theta}[\hat{\Pi}]=\frac{v_{33}[\hat{\Pi}]}{v_{33}[\hat{\Pi}]-g_{\theta}^{33}},
\ee
and $V_2=[v_{ij}]_{i,j\in\{1,2\}}$ in $D(3)$ is a $2\times2$ block matrix. 
The inclusion $\tilde{D}_{HGM}\subset D_{HGM}$ is trivial from the definition. 
With the same line of logic as before, we obtain the following result:
\begin{proposition}
$D(3)=\tilde{D}_{HGM}\subset D_{HGM}$ holds. 
\end{proposition}

Consequences of the above result are emphasized here: 
First, even though the value of $\theta_3$ is unknown, the structure of 
the MSE region $D(3)$ is the same as the previous region $D$ for the two-parameter case. 
The change solely enters as the scaling factor $\gamma_{\theta}$ which depends on 
MSE of $\theta_3$, i.e., $v_{33}[\hat{\Pi}]$, and this factor is 
strictly larger than $1$. 
This implies the relation 
\begin{multline}
D\subsetneq D_2(3):=\{V\in M_h| \det(V- \gamma_{\theta} G_{\theta}^{-1})\ge \det (\gamma_{\theta} G_{\theta}^{-1}),\\ V>\gamma_{\theta} G_{\theta}^{-1}\},
\end{multline}
%\be
%D\subset D_2(3)=\{V\in M_h| \det(V- \gamma_{\theta} G_{\theta}^{-1})\ge \det (\gamma_{\theta} G_{\theta}^{-1}),\,V>\gamma_{\theta} G_{\theta}^{-1}\},
%\ee
for each given value of the error $v_{33}[\hat{\Pi}]$. 
Second, the trade-off between errors in $\theta_I=(\theta_1,\theta_2)$ and $\theta_3$ is 
understood. The smaller error in $\theta_3$ gives the larger $\gamma_{\theta}$ 
resulting in large error in $\theta_I$. We then wish to make $v_{33}$ so large 
that $\gamma_{\theta}\simeq 1$. But, this means that we cannot 
perform the optimal POVM \eqref{optPOVM2} precisely. This kind of 
trade-off is typical in quantum theory and we think the general formalism 
to deal with effects of nuisance parameters in quantum estimation theory 
deserves further studies. 
If the SLD CR bound is used in stead of the HGM bound, 
we have the following MSE region obtained form the SLD CR bound: 
\be
D_{SLD}(3)=\{V\in M_{h}^{(3)} \,|\, V_2\ge G_{\theta}^{-1},\ v_{33}\ge g_{\theta}^{33}\}.   
\ee
This is different from the MSE matrix allowed by quantum mechanics.  
In particular, $D(3) \subsetneq D_{SLD}(3)$ holds. This shows that 
one should analyze achievable bound when considering the effect of nuisance parameters in general. 

\subsection{Achievability of the bound for no nuisance parameter}
In this subsection, we discuss achievability of the bound \eqref{nhgm2} which was 
derived for the case of no nuisance parameter. 
In particular, we show that the above bound with the nuisance parameter 
provides the same bound in the asymptotic limit. 

This is a direct consequence of simple structure of MSE region $D(3)$. 
It is well-known that the additivity 
of SLD Fisher information gives $G_{\theta}(3) \to n G_{\theta} (3)$ for 
the $n$th i.i.e. extended model $\cM_Q=\{\rho_\theta^{\otimes n}\,|\,\theta\in\Theta \}$. 
Let us consider an estimation strategy in which we use a fraction of 
$n$ states, say $\sqrt{n}$, to estimate $\theta_3$ and use the rest 
$n-\sqrt{n}$ states to estimate the relevant parameters $\theta_1,\theta_2$. 
Since the MSEs scales as $V_{\theta_1,\theta_2}\propto (n-\sqrt{n})^{-1} $ 
and $v_{33}\propto n^{-1/2}$, we see that the factor $\gamma_\theta$ 
scales as $\gamma_{\theta}\simeq 1$ for sufficiently large $n$. 
Therefore, the MSE region for $V_{\theta_1,\theta_2}$ converges  
to that of no nuisance parameter in $n\to\infty$ limit. 

To translate the above picture into more formal language, let us consider 
the $n$th i.i.e. extended model and consider an estimator with separable POVMs. 
Denoting the MSE matrix for this extended model as $V_\theta[\hat{\Pi}^n_{\mathrm{sep}}]$, 
the relation \eqref{qcmonotone} and the general bound \eqref{mibound} provide, 
\be
V_\theta[\hat{\Pi}^n_{\mathrm{sep}}]\ge \frac{1}{n}(J_\theta[\Pi])^{-1}\ge \frac{1}{n} C^{\mathrm{MI}}_\theta[W]. 
\ee
Let us write the rescaled MSE matrix as 
$V_\theta[\hat{\Pi}^n_{\mathrm{sep}}]\simeq \overline{V_\theta}[\hat{\Pi}^n_{\mathrm{sep}}]/n$, 
the above inequality and the HGM bound for three parameters, i.e., with the nuisance parameter, 
give us, $ \det(\overline{V_{\theta_I}}[\hat{\Pi}^n_{\mathrm{sep}}]
 -\overline{\gamma_{\theta}} G_{\theta}^{-1})\ge \det (\overline{\gamma_{\theta}} G_{\theta}^{-1})$, 
where $\overline{V_{\theta_I}}$ is the rescaled $2\times 2$ MSE matrix for $\theta_I=(\theta_1,\theta_2)$ and 
$\overline{\gamma_{\theta}}= {\overline{v_{33}}[\hat{\Pi}]}/{(\overline{v_{33}}[\hat{\Pi}]-g_{\theta}^{33})}$ 
is the rescaled factor. If we apply the considered estimation strategy, 
we have
\begin{align}\nonumber
& \det\left(\frac{n}{n-\sqrt{n}}\overline{V_{\theta_I}}[\hat{\Pi}^n_{\mathrm{sep}}]
 -\overline{\gamma_{\theta}} G_{\theta}^{-1}\right)\ge \det (\overline{\gamma_{\theta}} G_{\theta}^{-1}), \\
& \overline{\gamma_{\theta}}= \frac{\overline{v_{33}}[\hat{\Pi}]}{\overline{v_{33}}[\hat{\Pi}]-n^{-1/2}g_{\theta}^{33}}
\end{align}
We thus see that the MSE matrix $\overline{V_{\theta_I}}[\hat{\Pi}^n_{\mathrm{sep}}]$ 
for the relevant parameters satisfies 
\be
\det(\overline{V_{\theta_I}} -G_{\theta}^{-1})\ge \det (G_{\theta}^{-1}),
\ee
in the $n\to\infty$ limit. 

We next discuss more efficient way to achieve the previous bound 
based on the optimal POVM for the Nagaoka bound \eqref{optPOVM2}.  
Given sufficiently large $n$ copies of quantum states $\rho_{\theta}$, 
we split $n$ into $\sqrt{n}$ and the rest $n-\sqrt{n}$. 
Let us use the first group to estimate the nuisance parameter $\theta_3$ and 
let the MSE be $v_{33}=c_{33} n^{-1/2}$. With this precision, we use the remaining 
$n-\sqrt{n}$ states to estimate $\theta_I=(\theta_1,\theta_2)$ with the optimal estimator 
described by \eqref{optPOVM2}. The limit $n\to\infty$ then leads to the bound \eqref{nhgm2}. 

To see this argument quantitatively, let us assume that the true value 
for $\theta_3$ is $\theta^*_3$. We first make an estimate  as $\theta^*_3+\delta \theta_3$ with 
$\delta \theta_3$ a standard deviation ($\simeq$ square root of MSE). 
With this estimate let us perform the POVM of the form \eqref{optPOVM2}. 
We note that error in $\theta_3$ solely enters in the matrix \eqref{ematrix} 
and the straightforward calculation shows the effect of this deviation 
gives rise to the change of parameters: 
\be
(\theta_1,\theta_2)\to(\theta_1\cos\delta\theta_3,\theta_2).
\ee 
Therefore, the classical Fisher information matrix about this sub-optimal measurement 
outcomes with this error in $\theta_3$ is expressed as 
\be
 \Delta_\theta J_{\theta}[\Pi_{opt}]\, \Delta_\theta \mbox{ with } 
 \Delta_\theta=\left(\begin{array}{cc}  \cos\delta\theta_3&0 \\0 & 1 \end{array}\right).
\ee
Clearly, for small error $\delta\theta_3\simeq c_{33}/\sqrt{n}$ we can 
approximate $ \cos\delta\theta_3\simeq 1- c_{33}^2/2n$ and 
this decreases faster enough to conclude that the Nagaoka bound \eqref{nhgm2} 
can be achieved at each $\theta_I=(\theta_1,\theta_2)$ asymptotically for a given weight matrix $W$. 

\section{Asymptotic bound: Holevo bound}\label{sec4}
In this section we shall discuss the Holevo bound for our parametric model. 
As stated in Section \ref{sec2}, the Holevo bound can be achieved by a 
collective POVM $\hat{\Pi}^n$ acting on $\rho^{\otimes n}$ in the asymptotic limit. 
We will see that the Holevo bounds are different whether the phase parameter 
is known or not. 

We first list the inverse of SLD and RLD Fisher information matrix 
for the model. 
When the phase parameter is completely known the model 
is two-dimensional and we have
\be \label{2sldrld}
G_{\theta}^{-1}
=\left(\begin{array}{cc}1-\theta_1^2 & -\theta_1\theta_2 \\[.5ex] 
 -\theta_1\theta_2  & 1-\theta_2^2\end{array}\right),\quad
\tilde{G}_\theta^{-1}=(1-s_\theta^2)\left(\begin{array}{cc}1 & 0 \\[0.5ex]
0 & 1\end{array}\right). 
\ee
%\begin{align} \nonumber
%G_{\theta}^{-1}
%&=\left(\begin{array}{cc}1-\theta_1^2 & -\theta_1\theta_2 \\[.5ex] 
% -\theta_1\theta_2  & 1-\theta_2^2\end{array}\right),\\
%\tilde{G}_\theta^{-1}&=(1-s_\theta^2)\left(\begin{array}{cc}1 & 0 \\[0.5ex]
%0 & 1\end{array}\right). 
%\end{align}
Therefore, the RLD Fisher is real and we easily see that the SLD Fisher is more informative 
than RLD Fisher information, i.e., $G_{\theta}^{-1}\ge \tilde{G}_\theta^{-1}$. 
When the phase parameter $\theta_3$ is not known precisely 
and needs to be estimated, 
the inverse matrices of two quantum Fisher information are 
\begin{align}\nonumber
G_{\theta}(3)^{-1}&=
\left(\begin{array}{ccc}1- \theta_1^2& -\theta_1\theta_2 &0  \\[.5ex] -\theta_1\theta_2 & 1-\theta_2^2 & 0 \\[.5ex]0 & 0 & 1/\theta_1^2\end{array}\right),\\
\tilde{G}_\theta(3)^{-1}&=
\left(\begin{array}{ccc}1- \theta_1^2& -\theta_1\theta_2 & -\I\theta_2/\theta_1 \\[.5ex]
 -\theta_1\theta_2 & 1-\theta_2^2 & \I \\[.5ex]
 \I\theta_2/\theta_1 & -\I & 1/\theta_1^2\end{array}\right). 
\end{align}
%\be
%G_{\theta}^{-1}(3)=
%\left(\begin{array}{ccc}1- \theta_1^2& -\theta_1\theta_2 &0  \\[.5ex] -\theta_1\theta_2 & 1-\theta_2^2 & 0 \\[.5ex]0 & 0 & 1/\theta_1^2\end{array}\right),\quad
%\tilde{G}_\theta^{-1}(3)=
%\left(\begin{array}{ccc}1- \theta_1^2& -\theta_1\theta_2 & -\I\theta_2/\theta_1 \\[.5ex]
% -\theta_1\theta_2 & 1-\theta_2^2 & \I \\[.5ex]
% \I\theta_2/\theta_1 & -\I & 1/\theta_1^2\end{array}\right). 
%\ee
It is easy to see that $\mathrm{Re}\,\{\tilde{G}_\theta(3)^{-1}\}=G_{\theta}(3)^{-1}$ and 
the difference of two matrices is neither positive nor negative, that is 
there is no ordering between $\tilde{G}_\theta(3)$ and $G_{\theta}(3)$. 

\subsection{No nuisance parameter case}
The Holevo bound can be evaluated by an optimization 
over the tangent space at $\theta$: $T_\theta=\mathrm{span}_{\bbr}\{ L_{\theta,1},L_{\theta,2},\dots,L_{\theta,k}\}$. 
A straightforward calculation shows that the Holevo bound coincides with the SLD CR bound. 
Alternate way to see this simple fact is as follows. Consider the cotangent vectors of SLD operators 
defined by
\be
L^j_\theta=\sum_{i=1}^2 (G_\theta^{-1})_{ij}L_{\theta,i}\ (j=1,2).
\ee
By inserting $\vec{X}=(L^1_\theta,L^2_\theta)=:\vec{L_\theta}$ in the Holevo function $h_\theta[X|W]$, 
we see that the imaginary part of the matrix $Z_\theta[\vec{X}]$ vanishes. 
In this case the Holevo function coincides with the SLD CR bound, i.e.,
\be
C_\theta^H[W]=h_\theta[\vec{L}_\theta|W]=\Tr{WG_\theta^{-1}}=C_\theta^S[W] . 
\ee
Using the simple fact $\Tr{WA}\ge 0,\ \forall W>0\Rightarrow A\ge0$ for any $k\times k$ real symmetric matrix, 
we arrive at rather remarkable result: For any locally unbiased estimator, its MSE matrix satisfies
\be
V_\theta[\hat{\Pi}]\ge G_\theta^{-1}, 
\ee
where the equality can be attained with a sequence of collective POVMs 
which are asymptotically unbiased in the $n\to\infty$ limit, 
i.e., $\lim_{n\to\infty}n V_\theta[\hat{\Pi}^n]=G_\theta^{-1}$. 
Correspondingly, the MSE region allowed by the Holevo bound is 
\be
D_H:=\{ V\in M_h\,|\, V\ge G_\theta^{-1}\}. 
\ee
This proves the SLD CR bound can be achievable in the asymptotic limit, 
even though two SLD operators do not commute. 

\subsection{Nuisance parameter case}
Any qubit model of estimating three parameters becomes 
D-invariant if all SLD operators are linearly independent. 
This is true for our case as well and the Holevo bound is identical to the RLD CR bound. 
Thus, we have
\be
C^H_\theta[W]=\Tr{W G_\theta(3)^{-1}}+\mathrm{TrAbs}\{W\mathrm{Im}\,\tilde{G}_\theta(3)^{-1}\}.
\ee
The second term can be simplified as follows. Let $A=W\mathrm{Im}\tilde{G}_\theta(3)^{-1}$ 
be the $3\times 3$ real matrix whose eigenvalues are to be calculated. 
This matrix $A$ has good symmetry which gives 
\be
\Tr{A^\ell}=0\ \mbox{for odd }\ell, 
\ee
and $\det{A}=0$. The Caley-Hamilton theorem gives that the eigenvalues of $A$ are 
$0$, $\pm\sqrt{\Tr{A^2}/2 }$. The second term of the Holevo bound is 
written as
\be
\mathrm{TrAbs}\{W\mathrm{Im}\,\tilde{G}_\theta(3)^{-1}\}=\sqrt{2\Tr{(W\mathrm{Im}\,\tilde{G}_\theta(3)^{-1})^2}}. 
\ee

If we set the weight matrix $W$ as the form of the block diagonal one \eqref{bdform}, 
the above term reads 
\begin{multline}\
\mathrm{TrAbs}\{W(3)\mathrm{Im}\,\tilde{G}_\theta(3)^{-1}\}\\
=2\sqrt{w_3g_\theta^{33}}\sqrt{\Tr{W(G_\theta^{-1}-\tilde{G}_\theta^{-1})}}.
\end{multline}
%\be
%\mathrm{TrAbs}\{W(3)\mathrm{Im}\,\tilde{G}_\theta^{-1}(3)\}
%=2\sqrt{w_3g_\theta^{33}}\sqrt{\Tr{W(G_\theta^{-1}-\tilde{G}_\theta^{-1})}}.
%\ee
Here $W$ is the $2\times 2$ block matrix and 
$G_\theta^{-1}$ and $\tilde{G}_\theta^{-1}$ are the inverse of SLD and RLD Fisher information 
matrix for the known phase case, i.e., Eqs.~\eqref{2sldrld}. 
The final form of the Holevo bound is
\begin{multline}\label{holevow3}
C^H_\theta[W(3)]=\Tr{W G_\theta^{-1}}+w_3g_\theta^{33}\\
+2\sqrt{w_3g_\theta^{33}}\sqrt{\Tr{W(G_\theta^{-1}-\tilde{G}_\theta^{-1})}}. 
\end{multline}
%\be\label{holevow3}
%C^H_\theta[W(3)]=\Tr{W G_\theta^{-1}}+w_3g_\theta^{33}+2\sqrt{w_3g_\theta^{33}}\sqrt{\Tr{W(G_\theta^{-1}-\tilde{G}_\theta^{-1})}}. 
%\ee
By analyzing $\Tr{W(3)V_\theta}\ge C^H_\theta[W(3)]$ for all $W(3)>0$ as before, we obtain 
the MSE region allowed by the Holevo bound as
\begin{align}\nonumber
D_H(3)=\{ V\in M_{+}(3)\,|\,
V_2&\ge\gamma_\theta G_\theta^{-1}-(\gamma_\theta-1) \tilde{G}_\theta^{-1} ,\\
& V_2> G_\theta^{-1},\ v_{33}>g_\theta^{33}\},  
\end{align}
where $\gamma_\theta$ is defined by Eq.~\eqref{gamma1}. 
From this result, we see that the first term corresponds to the case of no nuisance 
parameter with the scaling factor $\gamma_\theta$. The second term, 
which is a negative matrix, represents non-trivial contribution from collective measurements.  
To see the structure of this Holevo bound, we rewrite the right hand side as 
$V_2\ge G_\theta^{-1}+(\gamma_\theta-1)(G_\theta^{-1}- \tilde{G}_\theta^{-1})\ge G_\theta^{-1}$. 
The last matrix inequality follows from $\gamma_\theta>1$ and $G_\theta^{-1}\ge \tilde{G}_\theta^{-1}$. 
Clearly, this shows that the Holevo bound for two-parameter case with no nuisance parameter 
cannot be attained exactly. However, by the same argument as before, 
one can find a sequence of POVMs such that $\gamma_\theta\to 1$ in the asymptotic limit. 
The MSE matrix for relevant parameters $\theta_1,\theta_2$ behaves as $V_{\theta_I} \simeq G_\theta^{-1}/n$.

\subsection{Comparison and discussion}
From these above results, we can expect that the ultimate bound can be quite different 
in general for quantum estimation problem wether there are nuisance parameters or not. 
This is because the error in the nuisance parameters enter as the bound of 
the MSE matrix for the relevant parameters. 

The model studied in this paper is a very special one in the sense that 
quantum Fisher information and all bounds do not depend on the value of the phase parameter 
(nuisance parameter). When the precision bound depends on the nuisance parameter, 
one has to substitute a rough estimate or adopt the worst case in order to 
derive a reliable bound for MSE for the relevant parameters. 

We also point out that our model meets the orthogonal condition with respect the SLD Fisher information. 
This orthogonality condition plays an important role in classical estimation problem 
and it guarantees the equality in \eqref{nuiineq}, that is, the bounds become same 
regardless whether there are nuisance parameters or not. In the quantum case, on the other hand, 
our result indicates that quantum version orthogonality condition itself does not 
conclude the same bound for the nuisance parameter case. In the first place, 
there are many quantum versions of Fisher information and we cannot 
say $\theta_I$ and $\theta_N$ are orthogonal in general. Indeed, it happens in our model that they 
are orthogonal with respect to the SLD Fisher information but not for the RLD Fisher information. 
Although our model is a very simple qubit system, it contains interesting and unique features of 
quantum parameter estimation problem. 

Let us briefly compare the bounds for separable POVMs and collective POVMs. 
In our model, the case of no nuisance parameter 
states that the Nagaoka bound \eqref{nhgm2}, which is truly greater than the SLD CR bound, 
can be improved significantly up to the SLD CR bound by collective measurements.  
When the value of phase $\theta_3$ is not known with infinite precision, 
the Holevo bound is strictly greater than the SLD CR bound. 
Collective POVMs improves the HGM bound, but we cannot reach the SLD CR bound for finite $n$. 
The analysis of this problem indicates: 
i) Importance of imaginary parts of RLD Fisher information and 
ii) proper treatment of nuisance parameters in quantum estimation problem. 

\section{Summary and outlook}\label{sec5}
In this paper, we have discussed a simple quantum two-dimensional parametric model 
with unknown (nuisance) parameter. It is clear that if the nuisance parameter is not 
orthogonal to the parameters of interest, one cannot ignore the effects of nuisance 
parameter in general. The case for unknown phase parameter was analyzed, 
and we have shown that the bound can be achieved asymptotically. 
More detailed asymptotic behavior of the optimal estimator with nuisance parameter 
should be studied as future work. 

The general structure for quantum parameter estimation theory with nuisance parameters needs to 
be explored. We do not know if orthogonal nuisance parameters can always be 
estimated similarly as was done in this paper. It is clear that non-orthogonal parameters cannot be 
estimated with the same error even in asymptotically. This conclusion directly follows from classical 
statistics particularly the general inequality \eqref{nuiineq}. However, in the quantum case, orthogonality 
condition does not guarantee efficient estimation as discussed in the previous subsection. 
This is also because optimal measurements in general 
depend on the unknown parameters and more detailed analysis needs to be involved. 

There are many important examples where the effects of nuisance parameters are 
important. An immediate application is quantum metrology in presence of 
unavoidable noises. The values of noises are not known with infinite precision by definition. 
Hence, one should take into account the fact that small errors in knowledge of noise parameters 
might spoil the efficient estimation which go beyond the classical precision scaling. 
These are largely unexplored territories and we shall make progress in due course. 

\section*{Acknowledgement}
The author is indebted to Prof.~Hiroshi Nagaoka for invaluable discussions and suggestions. 
He thanks Huangjun Zhu for providing information about the manuscript \cite{hjz14}. 

\section*{Appendix. Formulas}
We list useful formulas for computing SLD and RLD operators and the corresponding Fisher information 
based on the Bloch vectors. 
For a given qubit model, we can also regarded it as a model described by three dimensional 
real vector:
\be \label{qbmodel}
\cM_{\cB}=\left\{\v{s}_{\theta}\in\cB\, |\, \theta\in\Theta\subset{\mathbb R}^k \right\}.
\ee
Given a quantum statistical model \eqref{qbmodel}, SLD and RLD operators are expressed as 
\begin{align} \label{sldbrep}
L_{\theta,i}&=-\frac{\vecin{\del_i \stheta}{\stheta}}{1-s_\theta^2} \sigma_0+\left(\del_i\stheta
+ \frac{\vecin{\del_i \stheta}{\stheta}}{1-s_\theta^2}\stheta\right)\cdot \v\sigma \\  \nonumber
%&=\frac{1}{1-s_\theta^2}\left\{-\vecin{\del_i \stheta}{\stheta} \sigma_0
%+\left[\del_i\stheta+ \stheta\times (\stheta\times\del_i\stheta)\right]\cdot \v\sigma\right\},\\
\tilde{L}_{\theta,i}&=\frac{1}{1-s_\theta^2}\left[-\vecin{\del_i \stheta}{\stheta} \sigma_0
+\left(\del_i\stheta-\I \stheta\times\del_i\stheta\right)\cdot \v\sigma\right].
\end{align}
SLD and RLD Fisher information matrices read rather simple expressions as 
\begin{align} 
g_{\theta,ij}&=\vecin{\del_i\stheta}{\del_j\stheta}
+\frac{\vecin{\del_i\stheta}{\stheta}\vecin{\stheta}{\del_j\stheta}}{1-s_\theta^2} \\ \nonumber
&=\frac{1}{1-s_\theta^2}\left(\vecin{\del_i\stheta}{\del_j\stheta}-\vecin{\del_i\stheta\times\stheta}{\del_j\stheta\times\stheta}  \right), \\
\tilde{g}_{\theta,ij}&=\frac{1}{1-s_\theta^2}\left(\vecin{\del_i\stheta}{\del_j\stheta}
+\I\vecin{\del_i\stheta\times\del_j\stheta}{\stheta}  \right). 
\end{align}
In the above expressions, $s_\theta^2\equiv \vecin{\stheta}{\stheta}$ denotes the 
square of the length of the Bloch vector $\stheta$. 

The Holevo bound can also be expressed in terms of Bloch vectors as follows. 
Let $T_{\theta,i}^{\bot}=\{\v x\in\bbr^3\,|\, \vecin{x}{\del_i \stheta}=0 \}$ 
be the orthogonal space to the $i$th derivative of the Bloch vector. 
A linear operator which satisfies $\tr{\rho_\theta X}=0$ and $\del_i\tr{\rho_\theta X}=0$ 
can be expressed as 
\be
X=-\vecin{\stheta}{\v x}\sigma_0+{\v x}\cdot{\v \sigma}\mbox{ with } \v x\in T_{\theta,i}^{\bot}, 
\ee
and an element of the set appeared in the definition \eqref{Xset} takes the form of 
$\vec{X}=(X^1,X^2,\dots,X^k)$ with 
\begin{align}\nonumber
X^i&=-\vecin{\stheta}{\v x}^i\sigma_0+{\v x}^i\cdot{\v \sigma},\\
\v{x}^i&\in\bigcap_{j\neq i} T_{\theta,j}^{\bot},\quad \vecin{\v{x}^i}{\del_i\stheta}=1.
\end{align}
Using this form of Bloch vector representation, the $Z_\theta[\vec{X}]$matrix reads
\begin{align}\nonumber
\mathrm{Re} \,z_\theta^{ij}[X]&=\vecin{\v x^i}{\v x^j}-\vecin{\v x^i}{\stheta}\vecin{\stheta}{\v x^j},\\
\mathrm{Im} \,z_\theta^{ij}[X]&=-\vecin{\v x^i\times\v x^j}{\stheta}.
\end{align}
As noted in the text, the Holevo bound coincides with the RLD CR bound when $k=3$ for any qubit system 
if all SLD operators are linearly independent. 
Thus, the Holevo bound for $k=2$ is of interest and needs to be analyzed. 
With straightforward calculations, we have
\begin{multline}
h_\theta[X|W]=\sum_{i,j=1}^2w_{ij}(\vecin{\v x^i}{\v x^j}-\vecin{\v x^i}{\stheta}\vecin{\stheta}{\v x^j}) \\
+2\sqrt{\det W}\big|\vecin{\v x^i\times\v x^j}{\stheta} \big|,
\end{multline}
%\be
%h_\theta[X|W]=\sum_{i,j=1}^2w_{ij}(\vecin{\v x^i}{\v x^j}-\vecin{\v x^i}{\stheta}\vecin{\stheta}{\v x^j}) 
%+2\sqrt{\det W}\big|\vecin{\v x^i\times\v x^j}{\stheta} \big|,
%\ee 
for a given weight matrix $W=[w_{ij}]_{i,j\in\{1,2\}}$. 
Note this is a quadratic form with respect to $\v x^i$ and the Holevo 
bound can be obtained by standard optimization procedure.

\end{document}